\documentclass[11pt]{article}  
\usepackage{bm}
\usepackage{amscd}
\usepackage{amsmath}
\usepackage{amssymb}
\usepackage{amsthm}
\usepackage{epsf}
\usepackage{graphicx}
\usepackage{mathtools}
\usepackage{enumerate}
\usepackage{color}
\usepackage{cleveref} 
\usepackage{caption}
\usepackage{subcaption}
\usepackage{xspace} 

\newtheorem{theorem}{Theorem}
\newtheorem{lemma}[theorem]{Lemma}
\newtheorem{corollary}[theorem]{Corollary}
\newtheorem{remark}[theorem]{Remark}

\newcommand{\Z}{\mathbb{Z}}

\newcommand{\ve}{\varepsilon}
\newcommand{\eps}{\epsilon}

\newcommand{\lo}{\left( 1+ o(1) \right)}

\newcommand{\network}{graph\xspace}
\newcommand{\networks}{graphs\xspace}
\newcommand{\networkss}{graph's\xspace}

\newcommand{\Networks}{Graphs\xspace}

\newcommand\given[1][]{\:#1\vert\:}

\DeclarePairedDelimiter{\defaultDelim}{[}{]}

\DeclareMathOperator{\capPr}{\sf Pr}
\renewcommand{\Pr}[2][]{\capPr_{#1}\defaultDelim*{#2}}
\DeclareMathOperator{\capE}{\sf E}
\newcommand{\E}[2][]{\capE_{#1}\defaultDelim*{#2}}

\newcommand{\brac}[1]{\left(#1\right)}

\newcommand{\bfrac}[2]{\left(\frac{#1}{#2}\right)}
\crefformat{equation}{(#2#1#3)}

\DeclareMathOperator*{\argmax}{arg\,max}

\begin{document}
\title{Random Overlapping Communities: \\
Approximating Motif Densities of Large \Networks}
\author{ Samantha Petti\thanks{Georgia Tech, spetti@gatech.edu. 
Supported in part by an NSF graduate fellowship.}
\and Santosh Vempala \thanks{Georgia Tech, vempala@gatech.edu. Both authors were supported in part by NSF awards CCF-1563838 and CCF-1717349.}}
\maketitle

\begin{abstract}
A wide variety of complex networks  (social, biological, information etc.) exhibit local clustering with substantial variation in the clustering coefficient (the probability of neighbors being connected). Existing models of large \networks capture power law degree distributions (Barab\'asi-Albert) and small-world properties (Watts-Strogatz), but only limited clustering behavior. We introduce a generalization of the classical Erd\H{o}s-R\'enyi model of random \networks which provably achieves a wide range of desired clustering coefficient, triangle-to-edge and four-cycle-to-edge ratios for any given \network size and edge density. Rather than choosing edges independently at random, in the {\em Random Overlapping Communities} model, a \network is generated by choosing a set of random, relatively dense sub\networks (``communities''). We discuss the explanatory power of the model and some of its consequences. 
\end{abstract}

\section{Introduction}
Randomness has been an effective metaphor to model and understand the structure of complex networks. In 1959,  Erd\H{o}s and R\'enyi \cite{Erd59, Erd60} defined the simple random \network model $G_{n, p}$, where every pair of $n$ vertices is independently connected with probability $p$. Their seminal work transformed the field of combinatorics and laid the foundation of network science.  Mathematicians have extensively studied properties of \networks generated from this model and used it to prove the existence of \networks with certain properties. (See \cite{Fri15} for a survey.) The comparison of real-world \networks to $G_{n,p}$ is a popular tool for highlighting their nonrandom features \cite{Str01, New03, New05}. Moreover, the model has inspired more sophisticated random \network models, as predicted by  Erd\H{o}s and R\'enyi in the following remark from their pre-internet/pre-social \networks article:

{\em This may be interesting not only
from a purely mathematical point of view ...
if one aims at describing such a real situation,
one should replace the hypothesis of equiprobability of all connections by
some more realistic hypothesis.
It seems plausible that by considering the
random growth of more complicated structures 
one could
obtain fairly reasonable models of more complex real growth processes. 
}

 The two most influential random \network models designed to mimic properties of real-world \networks are the Watts-Strogatz {\em small world} model \cite{Wat98} and the Barab\'asi-Albert {\em preferential attachment} model \cite{Bar99}. Briefly, the first is a process that randomly rewires connections of a regular ring lattice \network. The resulting \networks have small diameter and high clustering coefficient (the probability that two neighbors of a randomly selected vertex are adjacent). The second is a growth model that repeatedly adds a new vertex to an existing \network and connects to existing vertices with probability proportional to their degree. This model exhibits and maintains a power law in the distribution of vertex degrees, another commonly observed phenomenon. 
 
These and other existing random \network models do not capture the following fundamental aspects of local structure: (1) Existing models cannot be tuned to produce \networks with arbitrary density, triangle-to-edge ratio, and four-cycle-to-edge ratio. (2) The clustering coefficients of \networks produced by existing models lie in very limited ranges determined by the \networkss density. In reality, the clustering coefficients of a variety of complex \networks (social, biological, information etc.) vary substantially and are not simply a function of the \networkss density \cite{New03}. 

We introduce the {\em Random Overlapping Communities (ROC)} model, a simple generalization of the Erd\H{o}s-R\'enyi model, which produces \networks with a wide range of clustering coefficients as well as triangle-to-edge and four-cycle-to-edge ratios. The model generates \networks that are the union of many relatively dense random communities. A {\em community} is an instance of $G_{s,q}$ on a set of $s$ randomly chosen vertices. A ROC \network is the union of many randomly selected communities that overlap, so  each vertex is a member of multiple communities.  The size $s$ and density $q$ of the communities determine clustering coefficient and triangle and four-cycle ratios. 

\paragraph{Capturing motif densities.}  A widely-used technique for inferring the structure and function of a \network is to observe overrepresented motifs, i.e., small patterns (sub\networks) that appear frequently. Recent work describes the overrepresented motifs of a variety of \networks including transcription regulation \networks, protein-protein interaction \networks, the rat visual cortex, ecological food webs, and the internet (WWW),  \cite{Yeg04, Alo07,Son05,Mil02}. The type of overrepresented motifs has been shown to be correlated with the \networkss function \cite{Mil02}. A model that produces \networks with high motif counts is necessary for approximating \networks whose function depends on the abundance of a particular motif. Here we focus on the two most basic motifs--- triangles and four-cycles.

A natural approach to constructing a \network with high motif density is to repeatedly add the motif on a randomly chosen subset of vertices. However, this process yields low motif to edge ratios for sparse \networks. For example, a \network on $n$ vertices with average degree less than $\sqrt{n}$ built by randomly adding triangles will have a triangle-to-edge-ratio at most 2/3. (See \Cref{just add tri}.) In \cite{New09} Newman considers a similar approach which produces \networks with varied degree sequences and triangle to edge ratio strictly less than 1/3. However, it is not hard to construct \networks with arbitrarily high triangle ratio (growing with the size of the \network). 

In the dense setting, a constant-size stochastic block model can be used to approximate \networks with high motif densities, as guaranteed by Szemer\'edi's regularity lemma (see \cite{Lov12}).  In a stochastic block model $M$, each vertex is assigned to one of $k$ classes, and an edge is added between each pair of vertices independently with probability $M_{i,j}$ where $i$ and $j$ are the classes of the vertices. However, the situation is drastically different for nondense \networks.
To construct a sparse \network with maximum degree at most $n^{1/3}$
with non-vanishing four-cycle density, the rank of $M$ must grow with the size of the \network. 
\begin{theorem}\label{block model}
	Let $M$ be a symmetric $n \times n$ matrix with entries in $[0,1]$ such that each row sum is at most $d$. Let $G$ be a graph on $n$ vertices obtained by adding each edge $(i,j)$ independently with probability $M_{ij}$. Then the expected number of $k$-cycles in $G$ at most $d^4 rank(M)$.
\end{theorem}
\noindent For example, the $d$-dimensional hypercube \network on $n=2^d$ vertices has a $\log(n)/4$ four-cycle-to-edge ratio; a stochastic block model $M$ that produces a \network of the same size, degree, and ratio must have rank at least $O(n/\log^2 n)$.

In contrast to the above approaches, the ROC model produces \networks with arbitrary triangle and four-cycle ratios independent of the density or size of the \network. In \Cref{real network bounds} we show that for almost all triangle and four-cycle ratios arising from some \network, there exists parameters for the ROC model to produce \networks with these ratios, {\em simultaneously}. Moreover, the vanishing set of triangle and four-cycle ratio pairs not achievable exactly can be approximated to within a small error. 

\paragraph{Clustering coefficient.} The clustering coefficient at a vertex $v$ is the probability two randomly selected neighbors are adjacent: $$C(v)= \frac{|\{\{a,b\} : a, b \in N(v), a\sim b\}|}{deg(v)(deg(v)-1)/2}.$$ Equivalently the clustering coefficient is twice the ratio of the number of triangles containing $v$ to the degree of $v$ squared. The ROC model is well suited to produce random \networks that reflect the high average clustering coefficients of real world \networks. Figure \ref{ccbar} illustrates the markedly high clustering coefficients of real-world \networks as compared with Erd\H{o}s-R\'enyi (E-R) \networks of the same density.  In \Cref{bounds}, we prove the average clustering coefficient of a ROC \network is approximately $sq^2/d$, meaning that tuning the parameters $s$ and $q$ with $d$ fixed yields wide range of clustering coefficients for a fixed density. Furthermore, \Cref{degree cc} describes the inverse relationship between degree and clustering coefficient in ROC \networks, a phenomena observed in protein-protein interaction \networks, the internet, and various social \networks \cite{Ste05, Mah06, Mis07, Yon07}.

\begin{figure}
	\centering
	\includegraphics[scale=.7]{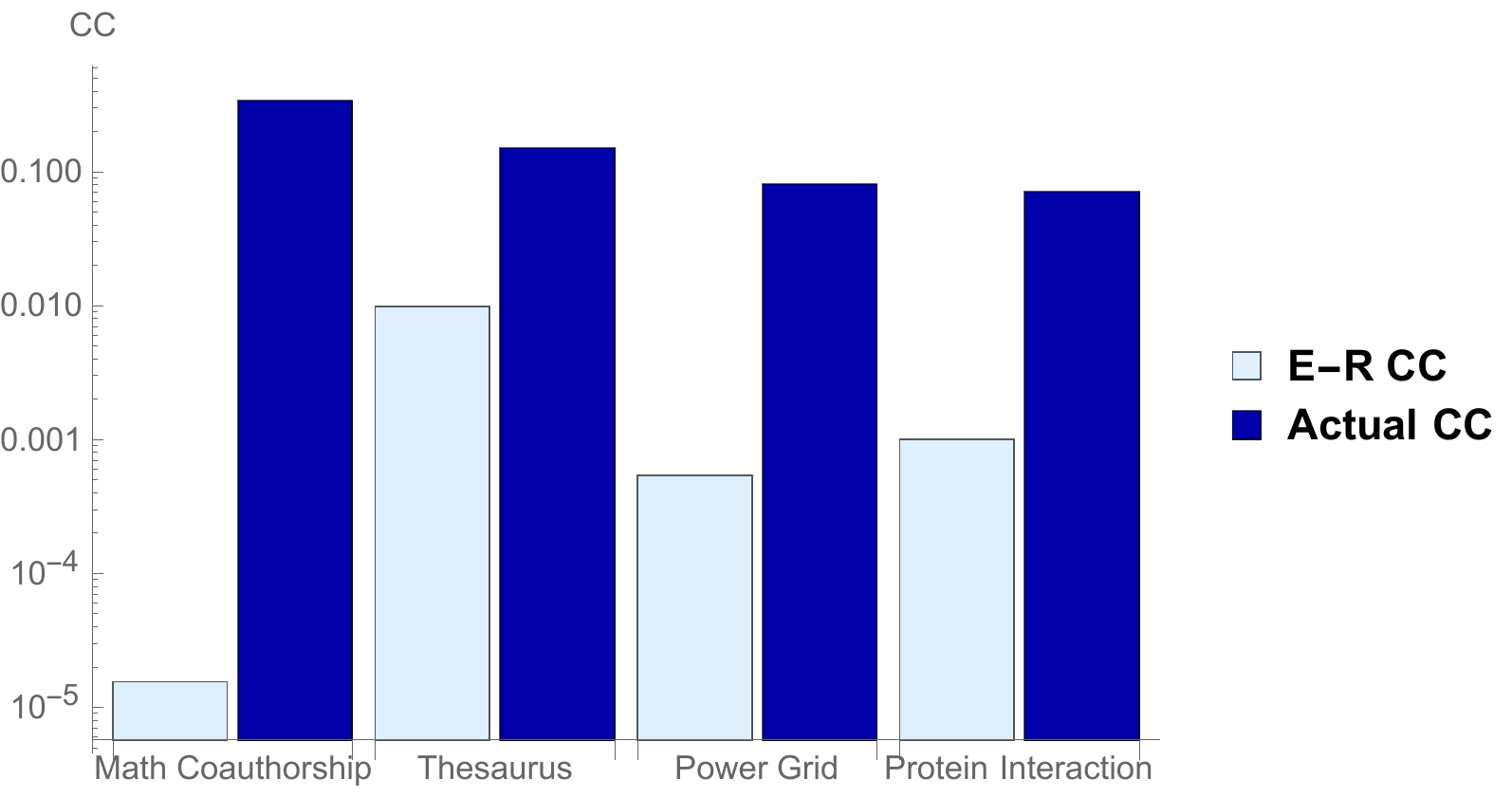}
	\caption{The clustering coefficient in real world \networks is much greater than that of an E-R random graph of the same density. Data from Table 3.1 of \cite{New03}.} \label{ccbar}
\end{figure}

\paragraph{Structure of the paper.} In \Cref{basic model} we introduce the ROC model, and then in \Cref{properties} we show the model's ability to produce \networks with specified size, density, triangle and four-cycle ratios and clustering coefficients. In \Cref{extension} we introduce a variation of the ROC model which produces \networks with various degree distributions and tunable clustering coefficient. We end with a discussion of the model's mathematical interest and explanatory value in real-world settings in Section \ref{discussion}.

\section{The Random Overlapping Communities model }\label{basic model}

A complex \network is modeled as the union of relatively dense, random communities. More precisely, to construct a \network on $n$ vertices with expected degree $d$, we pick $dn/(qs(s-1))$ random \networks, each of density $q$ on a random subset of $s$ of the $n$ vertices. 

\bigskip

\begin{centering}
\fbox{\parbox{0.95\textwidth}{
{\bf ROC($n,d,s,q$)}.  

Output: a \network on $n$ vertices with expected degree $d$.\\

Repeat $dn/(qs(s-1))$ times:
\begin{enumerate}
\item Pick a random subset $S$ of vertices (from $\{1,2,\ldots,n\}$) by selecting each vertex with probability $s/n$.
\item Add the random \network $G_{|S|,q}$ on $S$, i.e., for each pair in $S$, add the edge between them independently with probability $q$; if the edge already exists, do nothing.   
\end{enumerate}
}
}
\end{centering}

\bigskip

\begin{figure} 
\centering
\includegraphics[scale=.6]{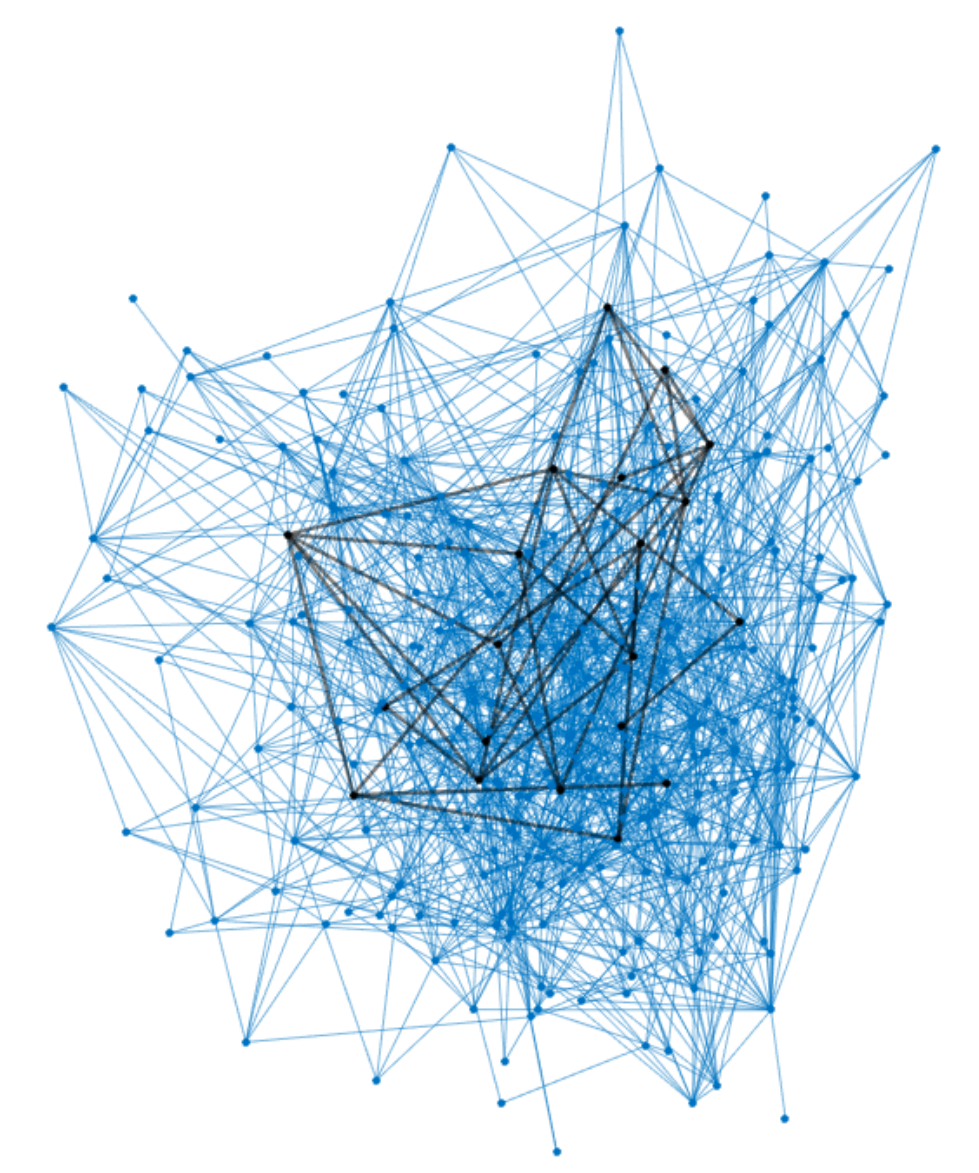}
\caption{In each step of the construction of a ROC($n,d,s,q$) \network, an instance of $G_{s,q}$ is added on a set of $s$ randomly selected vertices.}
\end{figure}

This generalizes the standard E-R model, which is the special case when $s=n$ and a single community is picked. For $G \sim ROC(n,d,s,q)$ the expected degree of each vertex is $d$.  If $d > sq \log{n}$ then with high probability $G$ will be connected. Moreover if $d/p> \log{\frac{nd}{s(s-1)p}}$, then with high probability the communities of $G$ will be connected even though there may be isolated vertices. See Section \ref{connectivity section} of the appendix for a further exploration of the connectivity properties of the ROC model.

\section{Approximation by ROC \networks}\label{properties}
In this section we analyze small cycle counts and local clustering coefficient of ROC \networks. For proofs of the theorems refer to Section \ref{proofs} of the appendix. 
We state our results as they hold asymptotically with respect to $n$.  

\subsection{Triangle and four-cycle count in ROC \networks.}
Define $R_k$ as the ratio between the number of $k$ cycles and the edges in a \network: $$R_k(G)=\frac{C_k(G)}{|E(G)|},$$ where $C_k(G)$ denotes the number of $k$ cycles in $G$. For $G \sim ROC(n,d,s,q)$, we instead define $$\overline{R}_k(G)=\frac{2\E{C_k(G)}}{nd},$$ the ratio of the expected number of $k$ cycles to the expected number of edges. 

\begin{lemma} \label{triangle} 
	Let $G \sim ROC(n,d,s,q)$  and $s = \omega(1)$. Then $$ \lim_{n \to \infty} \overline{R}_3(G)= \frac{sq^2}{3}  \mbox{ for } d=o(\sqrt{n}) \quad  \text{ and } \quad \lim_{n \to \infty} \overline{R}_4(G)= \frac{s^2q^3}{4} \mbox{ for } d=o(n^{1/3}).$$ 
	\end{lemma}

By varying $s$ and $q$, we can construct a ROC \network that achieves any ratio of triangles to edges or any ratio of four-cycles to edges. By setting $s=\sqrt{\log(n)}/4$ and $q=1$, we obtain a family of \networks with the hypercube four-cycle-to-edge ratio $\log(n)/4$, something not possible with any existing random \network model. 

Moreover, it is possible to achieve a given ratio by larger, sparser communities or by smaller, denser communities. For example communities of size 50 with internal density 1 produce the same triangle ratio as communities of size 5000 with internal density 1/10. Figure \ref{sq vary} illustrates the range of $s$ and $q$ that achieve various triangle and four-cycle ratios. Note that it is possible to achieve $R_3=3$ and $R_4 \in \{ 100,50,25\}$ but not $R_3=3$ and $R_4 \in \{3, 10\}$.

\begin{figure}[h]
\centering
\begin{subfigure}{.5\textwidth}
  \centering
  \includegraphics[scale=.4]{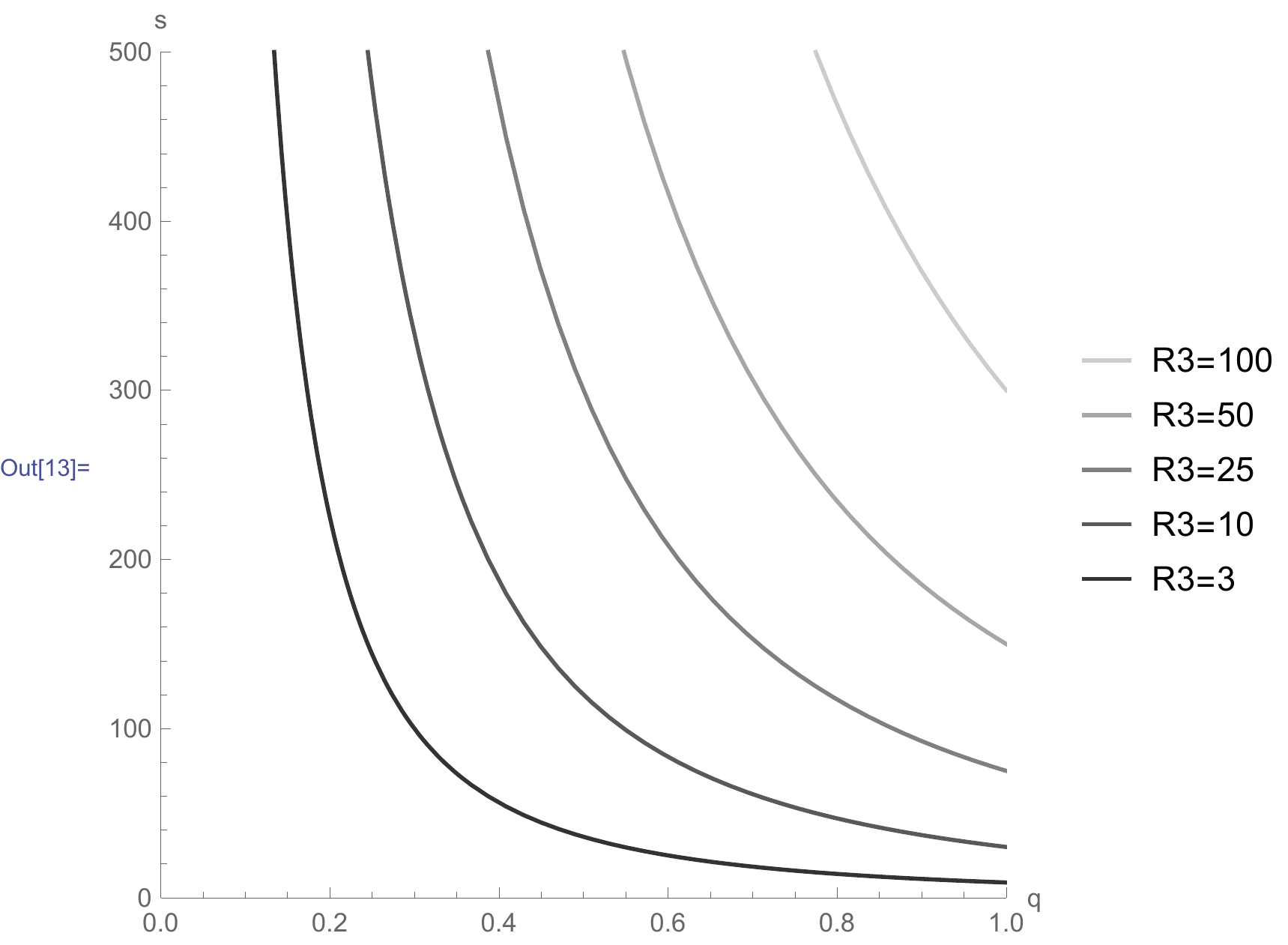}
  \label{fig:sub1}
\end{subfigure}%
\begin{subfigure}{.5\textwidth}
  \centering
  \includegraphics[scale=.4]{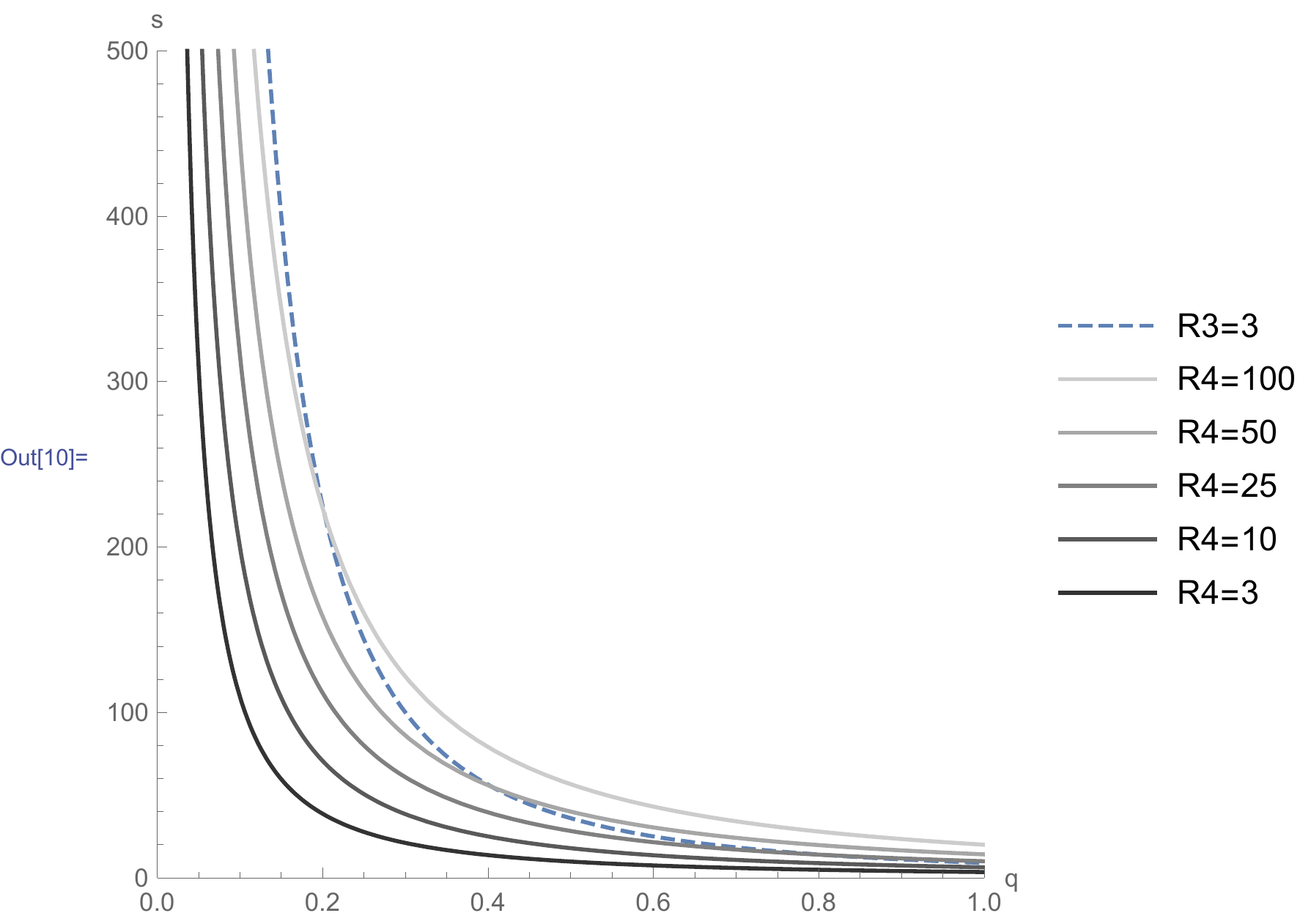}
  \label{fig:sub2}
\end{subfigure}
\caption{Left: A wide range of $s$ and $q$ yield the same $R_3$ and $R_4$ ratio (left and right respectively). }\label{sq vary}
\end{figure}

Next, we show that for almost all achievable pairs of triangle and four-cycle ratios, there exists a ROC construction that matches both ratios asymptotically. 

\begin{theorem}\label{real network bounds} The ROC model approximates most pairs of triangle and four-cycle ratios. 	
\begin{enumerate}
\item If there exists a \network $H$ with $R_3(H)=r_3$ and $R_4(H)=r_4$, then $3r_3(3r_3-1) \le 4r_4$. 
\item \label{match} For any $r_3$ and $r_4$ such that $9r_3^2\le 4r_4$, and $d= o(n^{1/3})$, the random \network\\
 $G \sim ROC\left(n,d, \frac{16r_4^2}{27r_3^3}, \frac{9r_3^2}{4r_4}\right)$ has $$\lim_{n \to \infty} \overline{R}_3(G)= r_3 \quad \text{ and } \quad  \lim_{n \to \infty} \overline{R}_4(G)= r_4.$$
\end{enumerate}
\end{theorem}

For every \network with triangle and four-cycle ratios in the narrow range $3r_3(3r_3-1) \le 4r_4 \le 9r_3^2$, there exists a ROC construction that matches $r_3$ and can approximate $r_4$ by $9r_3^2$, i.e., up to an additive error $3r_3/4$ (or multiplicative error of at most $1/(3r_3 -1)$ which goes to zero as $r_3$ increases). 


\subsection{Clustering coefficient.}
\Cref{bounds} gives an approximation of the expected clustering coefficient when the degree and average number of communities per vertex grow with $n$. The exact statement is given in \Cref{exact} of Section \ref{proofs}, and bounds in a more general setting are given by \Cref{more general cc}.

\begin{theorem} \label{bounds} Let $C(v)$ denote the clustering coefficient of a vertex $v$ with degree at least 2 in a \network drawn from $ROC(n,d,s,q)$ with $d=o(\sqrt{n})$, $d<(s-1) q e^{sq}$, $d= \omega( sq \log\frac{nd}{s})$, $s^2q=\omega(1)$, and $sq =o(d)$. Then 
$$\E{C(v)}=\left(1+o(1) \right) \frac{s q^2}{d}.$$ 
\end{theorem}

Unlike in E-R \networks in which local clustering coefficient is independent of degree, higher degree vertices in ROC \networks have lower clustering coefficient. High degree vertices tend to be in more communities, and thus the probability two randomly selected neighbors are in the same community is lower.  Figure \ref{compare oc er} illustrates the relationship between degree and clustering coefficient, the degree distribution, and the clustering coefficient for two ROC \networks with different parameters and the E-R random \network of the same density.

\begin{theorem} \label{degree cc}
 	Let $C(v)$ denote the clustering coefficient of a vertex $v$ in a \network drawn from $ROC(n,d,s,q)$ with $d=o(\sqrt{n})$, $s=\omega(1)$ and $deg(v) \geq 2sq$. Then 
$$\E{C(v) \given deg(v)=r}= \frac{sq^2}{r} \left(1 +o_r(1) \right) $$
 	\end{theorem}

\begin{figure}
\centering
\includegraphics[width=\textwidth]{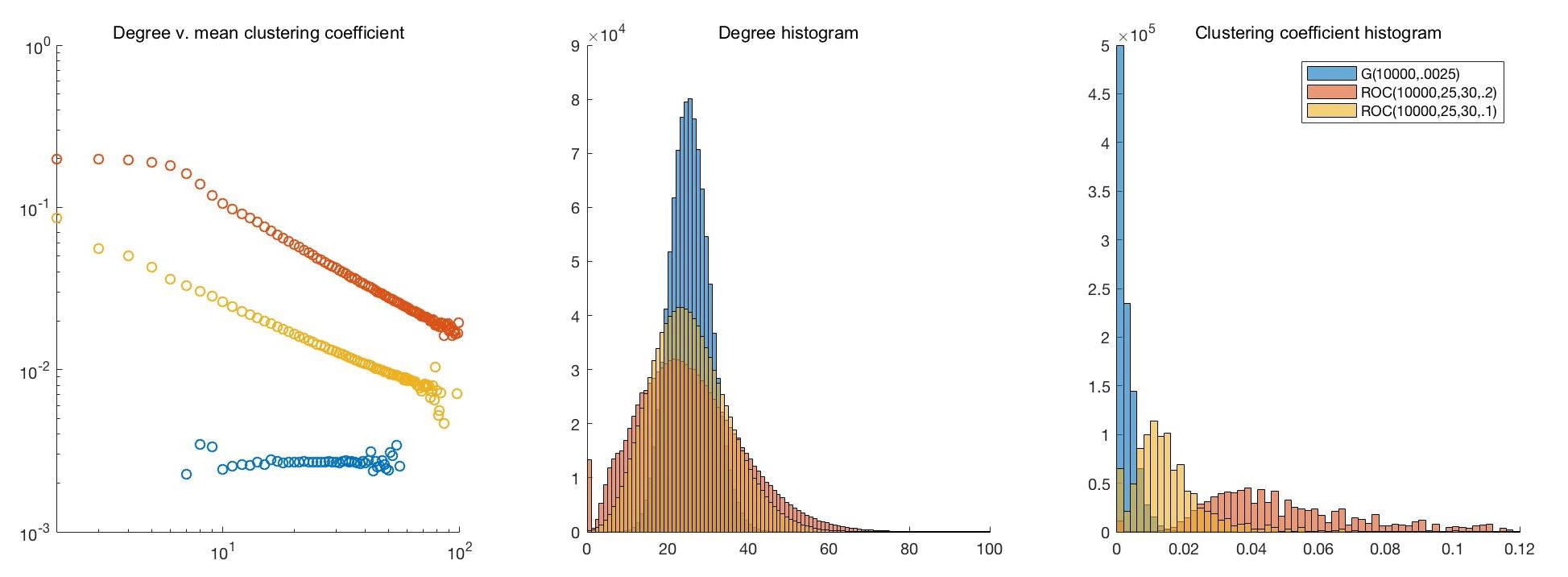}
\caption{A comparison of the degree distributions and clustering coefficients of 100 \networks with average degree 25 drawn from each $G_{10000,0.0025}$, $ROC(10000,25,30,0.2)$, and $ROC(10000,25,30,0.1)$. The mean clustering coefficients are $0.00270$, $0.06266$, and $0.01595$ respectively.  \label{compare oc er}}
\end{figure}

\begin{remark}  \label{the remark} The dependence between degree and clustering coefficient is the result of the variation in the numbers of communities a vertex is part of. To eliminate this variation and obtain a clustering coefficient distribution that is not highly dependent on degree, we can modify the $ROC$ construction as follows. Instead of selecting $s$ vertices uniformly at random to make up a community in each step, pre assign each vertex to precisely $\frac{d}{sq}$ communities of size $s$. In this setting the expected clustering coefficient can easily be computed: \begin{align*}\E{C(v)}= \Pr{\text{ two randomly selected nhbs are from the same community }}q
=\frac{sq^2}{d}.\end{align*}
Note also, that this variant of the ROC model will produce \networks with fewer isolated vertices. 
\end{remark}

\section{Diverse degree distributions and the DROC model }\label{extension}

In this section we introduce an extension of our model which produces \networks that match a target degree distribution in expectation. The extension is inspired by the Chung-Lu configuration model: given a degree sequence $d_1, \dots d_n$, an edge is added between each pair of vertices $v_i$ and $v_j$ with probability $\frac{d_i d_j}{\sum_{i=1}^n d_i}$, yielding a \network where the expected degree of vertex $v_i$ is $d_i$ \cite{Chu02}. In the  DROC model,  a modified Chung-Lu random \network is placed instead of an E-R random \network in each iteration. Instead of normalizing the probability an edge is selected in a community by the sum of the degrees in the community, the normalization constant is the expected sum of the degrees in the community. We use $D$ to denote a target degree sequence $t(v_1), \dots t(v_n)$, and $d$ to denote the mean. \bigskip

\begin{centering}
\fbox{\parbox{0.95\textwidth}{
{\bf DROC($n,D,s,q$)}.  

Output: a \network on $n$ vertices where vertex $v_i$ has expected degree $t(v_i)$.\\

Repeat $n/((s-1)q)$ times:
\begin{enumerate}
\item Pick a random subset $S$ of vertices (from $\{1,2,\ldots,n\}$) by selecting each vertex with probability $s/n$.
\item Add a modified C-L random \network on $S$, i.e., for each pair in $S$, add the edge between them independently with probability 
$ \frac{qt(v_i) t(v_j)}{sd} $; if the edge already exists, do nothing.   
\end{enumerate}
}
}
\end{centering}

\bigskip

\begin{theorem}\label{E1} Given a degree distribution $D$ with mean $d$ and $\max_i t(v_i)^2 \leq \frac{sd}{q}$,  DROC($n,D,s,q$) yields a \network where vertex $v_i$ has expected degree $t(v_i)$. \end{theorem}
 \noindent We require  $\max_i t(v_i)^2 \leq \frac{sd}{q}$ to ensure that the probability each edge is chosen is at most 1.

\begin{remark}
	Instead of requiring a sequence of $n$ target degrees as input to the DROC model, we can define the model with a distribution $\mathcal{D}$ of target degrees. In this altered version, Step 0 of the algorithm is to select a target degree for each vertex according to $\mathcal{D}$.
\end{remark}

\begin{remark} 
	Taking the distribution $D_d$ with $t(v)=d$ for all $v$ in the DROC model does not yield $ROC(n,d,s,q)$. The model $DROC(n, D_d, s,q)$ is equivalent to $ROC(n,d,s, \frac{qd}{s})$.
\end{remark}

The following corollary shows that it is possible to achieve a power law degree distribution with the DROC model for power law parameter $\gamma>2$. We use $\zeta(\gamma)=\sum_{n=1}^\infty n^{-\gamma}$ to denote the Riemann zeta function. 
\begin{corollary} \label{power law}
	Let $D\sim \mathcal{D_\gamma}$ be the power law degree distribution defined as follows: $$ \Pr{ t(v_i)=k}= \frac{k^{-\gamma}}{\zeta(\gamma)},$$ for all $1 \leq i \leq n$. If $\gamma>2$ and 
	$$\frac{s}{q} =\omega(1)\frac{\zeta(\gamma)}{\zeta(\gamma-1)} n^{\frac{1}{\gamma-1}},$$ then with high probability $D$ satisfies the conditions of \Cref{E1}, and therefore can be used to produce a DROC \network. 
\end{corollary}

\subsection{Clustering Coefficient.}
We show that by varying $s$ and $q$ we can control the clustering coefficient of a $DROC$ graph. 

\begin{theorem}\label{ex cc} Let $C(v)$ denote the clustering coefficient of a vertex $v$  in \network drawn from $DROC(n,D,s,q)$ with $\max t(v_i)^2 \leq \frac{sd}{q}$, $s= \omega(1)$, $s/n=o(q)$, and $t=t(v)$. Then  $$\E{C(v)} = \lo  \frac{\brac{\sum_{u \in V} t(u)^2}^2}{d^3n^2s} \brac{ (1- e^{-t})^2q^2 +c_t q^3},$$ where $c_t \in [0,6.2)$ is a constant depending on $t$.  \end{theorem}

\noindent Equation \cref{cc droc} in the proof of the theorem gives a precise statement of the expected clustering coefficient conditioned on community membership.

\section{Discussion and open questions} \label{discussion}

\paragraph{Modeling real-world \networks.}
The ROC model captures the degree distribution and clustering coefficient of graphs simultaneously. Previous work \cite{Hol02}, \cite{Ost13}, and \cite{Rav02} provides models that produce power law \networks with high clustering coefficients. Their results are limited in that the resulting \networks are restricted to a limited range of power-law parameters, and are either deterministic or only analyzable empirically. In contrast, the DROC model is a fully random model designed for a variety of degree distributions (including power law with parameter $\gamma >2$) and can  provably produce \networks with a wide range of clustering coefficient.  

Our model therefore may be a useful tool for approximating large \networks.  It is often not possible to test algorithms on \networks with billions of vertices (such as the brain, social \networks, and the internet). Instead, one could use the DROC model to generate a smaller \network with same clustering coefficient and degree distribution as the large \network, and then optimize the algorithm in this testable setting.  Further study of such a small \network approximation could provide insight into the structure of the large \network of interest. 

 Modeling a \network as the union of relatively dense communities has explanatory value for many real-world settings, in particular for social and biological networks. Social networks can naturally be thought of as the union of communities where each community represents a shared interest or experience  (i.e. school, work, or a particular hobby); the conceptualization of social networks as overlapping communities has been studied  in  \cite{Pal07}, \cite{Xie11}.
Protein-protein interaction networks can also be modeled by overlapping communities, each representing a group of proteins that interact with each other in order to perform a specific cellular process. Analyses of such networks show  proteins are involved in multiple cellular processes, and therefore overlapping communities define the structure of the underlying \network  \cite{Ahn10}, \cite{Kro06}, \cite{Bad03}. 

\paragraph{ROC vs mixed membership stochastic block models.} Mixed membership stochastic block models have traditionally been applied in settings with overlapping communities  \cite{Air08}, \cite{Kar11}, \cite{Air06}. The ROC model differs in two key ways. First, unlike low-rank mixed membership stochastic block models, the ROC model can produce sparse \networks with high triangle and four-cycle ratios. As discussed in the introduction, the over-representation of particular motifs in a \network is thought to be fundamental for its function, and therefore modeling this aspect of local structure is important. Second, in a stochastic block model the size and density of each community and the density between communities are all specified by the model. 
As a result, the size of the stochastic block model must grow with the number of communities, but the ROC model maintains a succinct description. This observation suggests the ROC model may be better suited for \networks in which there are many communities that are similar in structure, whereas the stochastic block model is better suited for \networks with a small number of communities with fundamentally different structures. Below we discuss extensions of the ROC model that maintain a succinct description and produce more diverse community structures. 

\paragraph{Open questions.}
\begin{enumerate}
	\item Consider the following extension. Instead of adding communities of size $s$ and density $q$, we define a probability distribution on a set of pairs $(s_i, q_i)$, and in each iteration choose a pair of parameters $(s_i, q_i)$ from the distribution and build the community $G_{s_i,q_i}$ on $s_i$ randomly selected vertices. Does this modification provide a better approximation for real-world graphs? 
	\item A further generalization involves adding particular sub\networks from a specified set according to some distribution instead of E-R graphs in each step (e.g., perfect matchings or Hamiltonian paths). Does doing so allow for greater flexibility in tuning the number of various types of motifs present (not just triangles and four-cycles)?
		\item A fundamental question in the study of graphs is how to identify relatively dense clusters. For example, clustering protein-protein interaction networks is a useful technique for identifying possible cellular functions of proteins whose functions were otherwise unknown \cite{Ste05, Kro06}. An algorithm designed specifically to identify the communities in a \network drawn from the ROC model has potential to become a state-of-the-art algorithm for clustering real-world networks with overlapping community structure. 
	\item The asymptotic thresholds for properties of E-R \networks have been studied extensively, see \cite{Fri15} for a survey. 
	Such questions are yet to be explored on ROC \networks, e.g., does every nontrivial monotone property have a sharp threshold? 
	\item How do graph algorithms behave on ROC graphs? For instance, what is the covertime of a random walk on a ROC graph?

\end{enumerate}

\bibliographystyle{plain}
\bibliography{networkrefs}

\newpage

\appendix
\section{Limitations of previous approaches}

\begin{theorem} \label{just add tri} Let $G$ be a \network on $n$ vertices obtained by repeatedly adding triangles on sets of three randomly chosen vertices. If the average degree is less than $\sqrt{n}$, the expected ratio of triangles to edges is at most 2/3. \end{theorem}

\begin{proof} 
Let $t$ be the number of triangles added and $d$ the average degree, so $d = 6t/n$. To ensure that $d < \sqrt{n}$, $t< n^{3/2}/6$. The total number of triangles in the \network is $t + (d/n)^3 {n \choose 3}= t +d^3/6= t+36t^3/n^3$. It follows that the expected ratio of triangles to edges is at most $$  \frac{t +36\left(\frac{t}{n}\right)^3}{3t}\leq  \frac{2}{3}.$$ 
\end{proof}

\begin{proof}(of \Cref{block model})	
Let $\sigma_1 \dots \sigma_{rank(M)}$ denote the eigenvalues of $M$.
\begin{align*} 
\E{\# k\text{-cycles} }&=\sum_{i_1 \not=i_2 \dots \not=i_k} M_{i_1i_2}M_{i_2 i_3}\dots M_{i_k i_1}\\
&\leq Tr(M^k)\\
&= \sum_{i=1}^{rank(M)}\sigma_i^k\\
&\leq rank(M) d^k.
\end{align*}
\end{proof}


\section{Connectivity of the ROC model}\label{connectivity section}
We describe the thresholds for connectivity for $ROC(n,d,s,q)$ networks. A vertex is isolated if it is has no adjacent edges. A community is isolated if it does not intersect any other communities. Here we use the abbreviation a.a.s. for asympotically almost surely. An event $A_n$ happens a.a.s. if  $\Pr{A_n} \to 1$ as $n\to \infty$. 

\begin{theorem}\label{no isolated vertices} For $ (s-1)q (\ln n+c)\leq d\leq (s-1)qe^{sq}(1-\ve)$, a \network from $ROC(n,d,s,q)$ a.a.s. has at most $\lo \frac{e^{-c}}{1-\ve}$ isolated vertices. \end{theorem}

\begin{proof} We begin by computing the probability a vertex is isolated,
\begin{align*} 
\Pr{v \text{ is isolated}  }&=   \sum_{i=0}^{\frac{nd}{s^2q}} \Pr{v \text{ is in $i$ communities} } (1-q)^{si}\\
&= \lo \sum_{i=1}^{\frac{nd}{s^2q}} { \frac{nd}{s(s-1)q} \choose i} \bfrac{s}{n}^i \brac{1-\frac{s}{n}}^{\frac{nd}{s(s-1)q} -i} e^{-sqi}\\
&\leq \lo e^{-\frac{d}{(s-1)q}}\sum_{i=0}^{\frac{nd}{s^2q}}  \bfrac{de^{-sq+\frac{s}{n}}}{(s-1)q}^ i   \\
&= \lo e^{-\frac{d}{(s-1)q}} \sum_{i=1}^{\frac{nd}{s^2q}}  \bfrac{de^{-sq}}{(s-1)q}^ i \\
&=\lo \brac{ e^{-\frac{d}{(s-1)q}}}\bfrac{1}{1-\ve}.
\end{align*}

Let $X$ be a random variable that represents the number of isolated vertices of a \network drawn from $ROC(n,d,s,q)$. We compute $$\Pr{X>0} \leq \E{X} = \lo n \brac{ e^{-\frac{d}{(s-1)q}}}\bfrac{1}{1-\ve}
 =\lo \bfrac{ e^{-c}}{1-\ve}.$$
\end{proof}

\begin{theorem} \label{no isolated communities} 
	A \network from  $ROC(n,d,s,q)$ with $s=o(\sqrt{n})$ has no isolated communities a.a.s. if 
	$$\frac{d}{q}> \log{\frac{nd}{s^2q}}.$$
\end{theorem}

\begin{proof} 
	We construct a ``community \network" and apply the classic result that $G(n,p)$ will a.a.s. have no isolated vertices when $p> (1+\eps)\log{n}/n$ for any $\eps>0$\cite{Erd59}. In the ``community \network" each vertex is a community and there is an edge between two communities if they share at least one vertex; a ROC \network has no isolated communities if and only if the corresponding ``community \network" is connected. The probability two communities don't share a vertex is $(1-\frac{s}{n})^s$. Since communities are selected independently, the ``community \network" is an instance of $G\left(\frac{nd}{s(s-1)q}, 1-(1-\frac{s}{n})^s\right)$. By the classic result, approximating the parameters by $\frac{nd}{s^2q}, 1-e^{s^2/n}$, this \network is connected when $$1-e^{-s^2/n}> \frac{\log{\frac{nd}{s^2q}}}{\frac{nd}{s^2q}}.$$ Since $s= o(\sqrt{n})$ is small, the left side of the inequality is approximately $s^2/n$, yielding the equivalent statement $$\frac{d}{q}>\log\frac{nd}{s^2q}.$$
\end{proof}

Note that the threshold for isolated vertices is higher, meaning that if a ROC \network a.a.s has no isolated vertices, then it a.a.s has no isolated communities. These two properties together imply the \network is connected.

\section{Section \ref{basic model} proofs}\label{proofs}

\begin{proof} 
	(of \Cref{triangle}.) Let $G \sim ROC(n,d,s,q)$ and $u,v \in V(G)$. First note that without information about whether $u$ and $v$ are in community together $\Pr{ u \sim v }=d/n=o(1)$ because each edge is equally likely. However, $\Pr{ u \sim v \given \text{ $u, v$ are in a common community } }=q+o(1)$. We show that both the triangle count and the four-cycle count are dominated by cycles contained entirely in one community. 
	
	We compute $\E{C_3(G)}$ by counting the total number of triangles in $G$. Let $T_1$ be the number triangles with all three edges originating in one community, $T_2$ be the number of triangles with two edges originating in the same community and the third edge originating in a different community, and $T_3$ be the number of triangles with edges originating in three different communities. We compute 
	
	\begin{align*} 
	\E{T_1}&= \left(\text{\# com.} \right)\E{\text{triangles in a com.}}=\frac{nd}{s(s-1)q}\frac{s^3q^3}{6}=\frac{ndsq^2}{6} \lo \\
	\E{T_2}&=\left(\text{\# com.} \right)\E{\# \text{two paths $u \sim v$, $v \sim w$ in a com.}}\Pr{\text{ $u \sim w$ in other com.}}\\
	&=\frac{nd}{s(s-1)q}s { s-1 \choose 2} q^2 \frac{d}{n}=\frac{d^2sq}{2} \lo \\
	\E{T_3}&=  \left(\text{\# triples $u,v,w \in V(G)$} \right)\Pr{\text{$u,v,w$ form a triangle}}={ n \choose 3} \left(\frac{d}{n}\right)^3= \frac{d^3}{6}.
	\end{align*}
	Therefore $$\E{R_3(G)}=\frac{2\E{T_1+T_2+T_3}}{nd}=\frac{sq^2}{3}(1+o(1)).$$ 

	Similarly, we compute $\E{C_4(G)}$ by summing over different categories of four-cycles based on the shared community membership of the vertices. For simplicity suppose the $a,b,c,d$ are the vertices of the four-cycle and let $C_1, \dots C_4$ denote different communities. If $\{a,b,c,d\} \in C_1$, the cycle is type 1. If $\{a,b,c\} \in C_1$ and $\{a,c,d\} \in C_2$, the the cycle is type 2. If $\{a,b, d\} \in C_1, \{b,c\} \in C_2, \{c,d\} \in C_3$, then the cycle is type 3. If $\{a,b \} \in C_1, \{b,c\} \in C_2, \{c,d\} \in C_3, \{d,a\} \in C_4$, then the cycle is type 4. Let $F_i$ be the number of cycles of type $i$.  We compute \begin{align*} 
	\E{F_1}&=\left(\text{\# com.} \right)\E{\text{\# four-cycles in a com.}}=\frac{nd}{s(s-1)q}\frac{3s^4q^4}{24}=\frac{nds^2q^3}{8}\lo\\
	\E{F_2}&=\E{\text{\# vertex pairs $u,v$ in two of the same coms. }}\left(\E{\text{\# common nhbs of $u,v$ in a com.}}\right)^2\\
	&={n \choose 2}  \bfrac{s}{n}^4  { \frac{nd}{s(s-1)q} \choose 2}\left( (s-2)q^2\right)^2= \frac{s^2q^2 d^2}{4} \lo\\
	\E{F_3}&= \left(\text{\# com.} \right)\E{\# \text{two paths $u \sim v$, $v \sim w$ in a com.}} |V(G)| \Pr{\text{ $x \sim w$ and $x \sim u$}}\\
	&=\frac{nd}{s(s-1)q} s { s-1 \choose 2} q^2\left( \frac{d}{n} \right)^2=\frac{sd^3}{2}\\
	\E{F_4}&=  \left(\text{\# quadruples $u,v,w,x \in V(G)$} \right)\E{\text{ways $u,v,w,x$ form a four-cycle}}\\
	&={ n \choose 4} 3 \left(\frac{d}{n}\right)^4= \frac{d^4}{8}.
	\end{align*}
	Therefore $$\E{R_3(G)}=\frac{2\E{F_1+F_2+F_3+F_4}}{nd}=\frac{s^2q^3}{3}(1+o(1)).$$ \end{proof}


\begin{proof} (of \Cref{real network bounds}.) 
(1) For each edge in $H$, let $t_e$ be the number of triangles containing $e$, so $\sum_{e \in E(H)} t_e = 3 C_3(H)= 3 r_3 |E(H)|$. If triangles $abc$ and $abd$ are present, then so is the four-cycle $acbd$. This four-cycle may also be counted via triangles $cad$ and $cdb$. Therefore $C_4(H) \geq \frac{1}{2} \sum_{ e\in E(H)} {t_e \choose 2}$. This expression is minimized when all $t_e$ are equal. We therefore obtain $$r_4|E(H)|=C_4(H)\geq \frac{|E(H)|}{2} {  3 r_3 \choose 2}  = \frac{3r_3(3r_3-1) |E(H)|}{4}.$$ It follows that $\frac{3r_3(3r_3-1)}{4r_4}\leq1$.  
 
 (2) Since the hypothesis guarantees $q \leq 1$, applying \Cref{triangle} to $G \sim ROC\left(n,d, \frac{16r_4^2}{27r_3^3}, \frac{9r_3^2}{4r_4}\right)$ implies the desired statements.  
\end{proof}


\begin{remark} \label{ignore}
\Cref{bounds} gives bounds expected clustering coefficient up to factors of $(1 + o(1))$.  The clustering coefficient at a vertex is only well-defined if the vertex has degree at least two. Given the assumption in  \Cref{bounds} that $d= \omega(sq \log{\frac{nd}{s}})$, $ d< (s-1)q e^{sq}$, and $s=\omega(1)$, \Cref{no d1} implies that the fraction of vertices of degree strictly less than two is $o(1)$. Therefore we ignore the contribution of these terms throughout the computations for \Cref{bounds} and supporting \Cref{exact}. In addition we divide by $deg(v)^2$ rather than by $deg(v)(deg(v)-1)$ in the computation of the clustering coefficient since this modification only affects the computations up to a factor of $(1 + o(1))$.
\end{remark}

\begin{lemma}\label{no d1} If $d= \omega( sq\log{\frac{nd}{s}})$, $s=\omega(1)$, $s=o(n)$, and $d<(s-1)q e^{sq}$, then a \network from $ROC(n,d,s,q)$ a.a.s. has no vertices of degree less than 2. \end{lemma}

\begin{proof}  
\Cref{no isolated vertices} implies there are no isolated vertices a.a.s.  
We begin by computing the probability a vertex has degree one. 
\begin{align*} \Pr{deg(v)=1 }&=  \sum_{i=1}^{\frac{nd}{s^2q}} \Pr{v \text{ is in $i$ communities} } q (1-q)^{si-1}\\
&=  \sum_{i=1}^{\frac{nd}{s^2q}} { \frac{nd}{s(s-1)q} \choose i} \bfrac{s}{n}^i \brac{1-\frac{s}{n}}^{\frac{nd}{s(s-1)q} -i} q (1-q)^{si-1}\\
&\leq \lo \sum_{i=1}^{\frac{nd}{s^2q}}  \bfrac{nd}{s(s-1)q}^ i \bfrac{s}{n}^i e^{-\frac{d}{sq}+\frac{si}{n}} q e^{-qsi+q}  \\
&= \lo q e^{-\frac{d}{sq}} \sum_{i=1}^{\frac{nd}{s^2q}}  \bfrac{de^{-sq}}{(s-1)q}^ i \\
&=O\brac{ \frac{de^{-sq-\frac{d}{sq}}}{s}}
\end{align*}
Let $X$ be a random variable that represents the number of degree one vertices of a \network drawn from $ROC(n,d,s,q)$. When $d= \omega( sq\log{\frac{nd}{s}})$, we obtain $$\Pr{X>0} \leq \E{X} = O\brac{ \frac{nde^{-sq-\frac{d}{sq}}}{s}}=o(1).$$
\end{proof}

\begin{lemma} \label{exact}
Let $C(v)$ denote the clustering coefficient of a vertex $v$ of degree at least 2 in a \network drawn from $ROC(n,d,s,q)$ with $d=o( \sqrt{n})$ and $d= \omega(sq \log{\frac{nd}{s}})$. Then $$\E{C(v)}=\left(1 + o(1) \right) \left( \sum_{i=1}^{\frac{nd}{s^2q}} { \frac{nd}{s^2q} \choose i} \left( \frac{s}{n}\right)^i \left(1- \frac{s}{n}\right)^{\frac{nd}{s^2q}-i}\frac{s(s-1)q^3k }{\left( sqk + 2-2q\right)^2} \right). $$ \label{cc}  
\end{lemma}

\begin{proof} 
For ease of notation, we ignore factors of $(1+ o(1))$ throughout as described in \Cref{ignore}.
First we compute the expected clustering coefficient of a vertex from an $ROC(n,d,s,q)$ \network given $v$ is contained in precisely $k$ communities. Let $X_1, \dots X_k$ be random variables representing the degree of $v$ in each of the communities, $X_i \sim Bin(s,q)$. We have \begin{align} \E{C(v)| \text{ $v$ in $k$ communities }}&=\E{ \frac{\sum_{i=1}^kX_i(X_i-1) q}{\left( \sum_{i=1}^k X_i\right)^2}} \label{given k} \\
&= qk \E{ \frac{X_1(X_1-1) }{\left( sq(k-1) + X_1\right)^2}} \nonumber \\
&= qk \E{ \frac{X_1^2 }{\left( sq(k-1) + X_1\right)^2}}- qk \E{ \frac{X_1 }{\left( sq(k-1) + X_1\right)^2}} \nonumber.
\end{align}
Write $X_1= \sum_{i=1}^s y_i$ where $y_i \sim Bernoulli(q)$. Using linearity of expectation and the independence of the $y_i's$ we have
\begin{align*} \E{ \frac{X_1 }{\left( sq(k-1) + X_1\right)^2}}&= s \E{ \frac{y_1 }{\left( sq(k-1) + (s-1)q+y_1\right)^2}}
&= \frac{sq }{\left( sq(k-1) + (s-1)q+1\right)^2},\\
\end{align*}
and 
\begin{align*} \E{ \frac{X_1^2 }{\left( sq(k-1) + X_1\right)^2}}&= \E{ \frac{\left(\sum_{i=1}^s y_i\right)^2 }{\left( sq(k-1) +\sum_{i=1}^s y_i\right)^2}}\\
&= s\E{ \frac{ y_1^2 }{\left( sq(k-1) +q(s-1) +y_1\right)^2}}+s(s-1)\E{ \frac{\left(y_1y_2\right)^2 }{\left( sq(k-1) + (s-2)q+y_1+y_2\right)^2}}\\
&= \frac{ sq }{\left( sq(k-1) +q(s-1) +1\right)^2}+\frac{s(s-1)q^2 }{\left( sq(k-1) + (s-2)q+2\right)^2}.
\end{align*}
Substituting in these values into \Cref{given k}, we obtain
\begin{equation}
\E{C(v)|v \in k \text{ communities }}=qk\left(\frac{s(s-1)q^2 }{\left( sq(k-1) + (s-2)q+2\right)^2}\right)=\frac{s(s-1)q^3k }{\left( sqk + 2-2q\right)^2}.\label{conditional}
\end{equation}

Let $M$ be the number of communities a vertex is in, so $M  \sim Bin\left(\frac{nd}{s^2q},\frac{s}{n}\right).$ It follows 
\begin{align*}
\E{C(v)}&=\sum_{i=1}^{\frac{nd}{s^2q}}\Pr{\text{ $v$ in $k$ communities }} \E{C(v)|\text{ $v$ in $k$ communities }}\\
&=\sum_{i=1}^{\frac{nd}{s^2q}} { \frac{nd}{s^2q} \choose i} \left( \frac{s}{n}\right)^i \left(1- \frac{s}{n}\right)^{\frac{nd}{s^2q}-i}\frac{s(s-1)q^3k }{\left( sqk + 2-2q\right)^2}.\\
\end{align*}
\end{proof}

The proof of \Cref{bounds}, relies on the follow two lemmas
regarding expectation of binomial random variables. 

\begin{lemma} \label{ex of rec p1} 
Let $X \sim Bin(n,p)$. Then 
\begin{enumerate}
\item $\E{ \frac{1}{X+1} \given X \geq 1 }= \frac{1- \left( 1-p\right)^{n+1}-(n+1)p(1-p)^n}{p(n+1)}$ and
\item $\E{ \frac{1}{X+1}}= \frac{1- \left( 1-p\right)^{n+1}}{p(n+1)}$.
\end{enumerate}
\end{lemma}

\begin{proof} 
Observe 
\begin{align*} 
\E{ \frac{1}{X+1} \given X\geq 1} &= \sum_{i=1}^n {n \choose i}  \frac{p^i (1-p)^{n-i}}{i+1}\\
&= \frac{1}{p(n+1)} \sum_{i=1}^n { n+1 \choose i+1} p^{i+1} (1-p)^{n-i}\\
&= \frac{1- \left( 1-p\right)^{n+1}-(n+1)p(1-p)^n}{p(n+1)}.
\end{align*}
Similarly 
\begin{align*} 
\E{ \frac{1}{X+1}} &= \sum_{i=0}^n {n \choose i}  \frac{p^i (1-p)^{n-i}}{i+1}
= \frac{1}{p(n+1)} \sum_{i=0}^n { n+1 \choose i+1} p^{i+1} (1-p)^{n-i}
= \frac{1- \left( 1-p\right)^{n+1}}{p(n+1)}.
\end{align*}
\end{proof}

\begin{lemma}\label{ex of rec}
 Let $X \sim Bin(n,p)$. Then $$\E{ \frac{1}{X} \given X \geq 1} \leq \frac{1}{p(n+1)}\left(1 +\frac{3}{p(n+2)} \right).$$
 \end{lemma}

\begin{proof} Note that when $X \geq 1$, $$ \frac{1}{X} \leq \frac{1}{X+1}+\frac{3}{(X+1)(X+2)}.$$ By Lemma \ref{ex of rec p1}, 
\begin{equation}\label{true} \E{\frac{1}{X+1}\given X \geq 1} \leq \frac{1}{p(n+1)}.
\end{equation}
We compute 
\begin{align*}
\E{\frac{1}{(X+1)(X+2)}\given X\geq 1}&= \sum_{i=1}^n \frac{{n \choose i} p^i (1-p)^{n-i}}{(i+1)(i+2)}\\
&= \frac{1}{p^2 (n+2)(n+1)} \sum_{i=1}^n {n +2 \choose i+2} p^{i+2}(1-p)^{n-i}\\
&\leq \frac{1}{p^2 (n+2)(n+1)}.
\end{align*}
Taking expectation of \Cref{true} gives $$\E{\frac{1}{X}\given X\geq 1} \leq \frac{1}{p(n+1)}\left(1 +\frac{3}{p(n+2)} \right).$$

 \end{proof}

\begin{proof} (of \Cref{bounds}.) For ease of notation, we ignore factors of $(1 + o(1))$, as described in \Cref{ignore}.
It follows from \Cref{conditional} in the proof of \Cref{exact} that
$$\frac{q}{k+1}\leq \E{C(v)|v \in k \text{ communities }}\leq \frac{q}{k}, $$ where the left inequality holds when $q(s-1)\geq 5$.

We now compute upper and lower bounds on $\E{C(v)}$, assuming $v$ is in some community. Let $M$ be the random variable indicating the number of communities containing $v$, $M \sim Bin\left(\frac{nd}{s(s-1) q}, \frac{s}{n}\right)$. It follows
\begin{align*}\E{C(v)}=\sum_{k=1}^{\frac{nd}{s^2q}}\Pr{M=k }\E{C(v)|M=k} \end{align*}
\begin{align*}q\E{\frac{1}{M+1} \given M \geq 1} \leq \E{C(v)}\leq  q\E{\frac{1}{M}\given M \geq 1}. \end{align*}
Applying Lemmas \ref{ex of rec p1} and \ref{ex of rec} to the lower and upper bounds respectively, we obtain 
$$\frac{q\left(1-\left(1-\frac{s}{n}\right)^{\frac{nd}{s(s-1)q}+1}- \brac{\frac{nd}{s(s-1)q}+1}\left(1-\frac{s}{n}\right)^{\frac{nd}{s(s-1)q}}\right)}{\frac{d}{(s-1)q}+\frac{s}{n}}\leq\E{C(v)}\leq \frac{q}{\frac{d}{(s-1)q} +\frac{s}{n}} \left( 1 + \frac{3}{\frac{d}{(s-1)q} +\frac{2s}{n}} \right) $$
which for $s=o(n)$ simplifies to
\begin{equation} \label{more general cc}
\lo \frac{(s-1)q^2}{d} \left( 1-\frac{nd}{s (s-1) q}e^{-d/((s-1)q)} \right) \leq \E{C(v) } \leq  \frac{(s-1)q^2}{d} \left( 1+ \frac{(s-1)q}{d}\right) \lo.
\end{equation}
Under the assumptions $s^2q =\omega(1)$ and $sq = o(d)$, we obtain our desired result $$\E{C(v)}=\left(1+o(1) \right)\left( \frac{sq^2}{d}\right).$$
\end{proof}

The following lemma will be used in the proof of \Cref{degree cc}.

\begin{lemma} \label{helper}
The $X$ be a nonnegative integer drawn from the discrete distribution with density proportional to $f(x)=x^{r-x}e^{-ax}$. Let $z=\argmax f(x)$. Then \[
\Pr{|x-z| \ge 2t\sqrt{z}} \le e^{-t+1}.
\]
\end{lemma}

\begin{proof}
First we observe that $f$ is logconcave:
\[
\frac{d^2}{dx^2}\ln f(x) = \frac{d}{dx}(-a + \frac{r}{x}-1 - \ln x) = -\frac{r}{x^2}-\frac{1}{x}
\]
which is nonpositive for all $x \ge 0$. We will next bound the standard deviation of this density, so that we can use an exponential tail bound for logconcave densities.
To this end, we estimate $\max f$. Setting its derivative to zero, we see that at the maximum, we have 
\begin{equation}\label{eq:max-cond}
a+1= \frac{r}{x}-\ln x.
\end{equation}  
The maximizer $z$ is very close to 
\begin{equation} \label{z}
\frac{r}{(a+1)+\ln\frac{r}{(a+1)+\ln(r/(a+1))}},
\end{equation}
and the maximum value $z$ satisfies
$z^{r-z}e^{-az}=z^r e^{-r+z}$. 
Now we consider the point $z+\delta$ where $f(z+\delta)=f(z)/e$, i.e., 
\[
\frac{(z+\delta)^{r-z-\delta}e^{-az-a\delta}}{z^{r-z}e^{-az}} \le e^{-1}.
\]
The LHS is 
\begin{align*}
\left(1+\frac{\delta}{z}\right)^{r-z}z^{-\delta}\left(1+\frac{\delta}{z}\right)^{-\delta}e^{-a\delta} &\le e^{\delta(\frac{r}{z}-1-a-\ln z)}e^{-\frac{\delta^2}{z}}\\ &\le e^{-\frac{\delta^2}{z}}
\end{align*}
where in the second step we used the optimality condition (\ref{eq:max-cond}). 
Thus for $\delta = \lo \sqrt{z}$, $f(x + \delta) \leq f(x)/e$.
By logconcavity\footnote{which says that for any $x,y$ and any $\lambda \in [0,1]$, we have $f(\lambda x +(1-\lambda)y)\ge f(x)^\lambda f(y)^{1-\lambda}$} we have 
$$f(x+ \delta)= f\brac{ \brac{ 1-\frac{1}{t} } x + \frac{1}{t} ( x +t\delta) } \geq f(x)^{1-1/t} f(x+t\delta)^{1/t}$$
for any $t \geq 1$. 
It follows 
\begin{equation} \label{f(x)}
f(x+t\delta) \le f(x)/e^t 
\end{equation}
 for all $t$ (since we can apply the same argument for $z-\delta$).
Taking $x=z$ in \Cref{f(x)} and using  the observation $\sum_{x \in \Z^{+} } f(x) \geq f(z)$, it follows that $$\Pr{ x= z+ t \sqrt{z}}\leq  e^{-t} \quad \text{ and } \quad \Pr{ x= z- t \sqrt{z}}\leq e^{-t}$$ and so $$\Pr{|x-z| \ge t\sqrt{z}} \le 2 e^{-t}\le e^{-t+1}.$$ 
\end{proof}

\begin{proof} (of \Cref{degree cc}). 
Let $M$ denote the number of communities a vertex $v$ is selected to participate in. We can write
\begin{align*}
\E{C(v)|deg(v)=r} &= \sum_{k=\frac{r}{s}}^r \E{C(v)|deg(v)=r, M=k}\Pr{M=k|deg(v=r}\\
&=\sum_{k=\frac{r}{s}}^r \E{C(v)|deg(v)=r, M=k}\Pr{deg(v)=r|M=k}\frac{\Pr{M=k}}{\Pr{deg(v)=r}}.
\end{align*}
First we compute the expected clustering coefficient of a degree $r$ vertex given that it is $k$ communities: 
\begin{align*}
\E{C(v)|deg(v)=r \text{ and  }M=k}=\frac{ \sum_{i \not= j, i,j \in N(v)} q\left(\Pr{i, j \text{ part of same community} } \right)}{deg(v)\left(deg(v)-1\right)} =\frac{q}{k}.
\end{align*}
Next we note that $M$ is a drawn from a binomial distribution, and the degree of $v$ is drawn from a sum of $k$ binomials, each being $Bin(s,q)$.	Therefore,
\begin{align*}
\Pr{M=k}\Pr{deg(v)=r|M=k} &= {\frac{nd}{s(s-1)q} \choose k} \left(\frac{s}{n}\right)^k\left(1-\frac{s}{n}\right)^{\frac{nd}{s(s-1)q}-k} {sk \choose r} q^r(1-q)^{sk-r}.
\end{align*}
Using this we obtain
\begin{align}
	\E{C(v)|deg(v)=r} &= \frac{\sum_{k=\frac{r}{s}}^r \frac{q}{k}\Pr{M=k}\Pr{deg(v)=r|M=k}}{\sum_{k=\frac{r}{s}}^r\Pr{M=k}\Pr{deg(v)=r|M=k}}\nonumber \\
	&=\lo q \frac{\sum_{k=\frac{r}{s}}^r \frac{1}{k}\cdot \left(\frac{d}{(s-1)qk}\right)^ke^{-\frac{d}{(s-1)q}+\frac{sk}{n}} \left(\frac{skq}{r}\right)^re^{-qsk+qr}}{\sum_{k=\frac{r}{s}}^r\left(\frac{d}{(s-1)qk}\right)^ke^{-\frac{d}{(s-1)q}+\frac{sk}{n}} \left(\frac{skq}{r}\right)^re^{-qsk+qr}}\nonumber\\
	&=\lo 	q \frac{\sum_{k=\frac{r}{s}}^r \frac{1}{k}\cdot \left(\frac{d}{(s-1)q}\right)^k k^{r-k}e^{-qsk}}{\sum_{k=\frac{r}{s}}^r
	\left(\frac{d}{(s-1)q}\right)^k k^{r-k}e^{-qsk}}. \label{expression}
\end{align}		
Writing $a=qs-\ln(d/(s-1)q)$, this is 
\[
q \frac{\sum_{k=\frac{r}{s}}^r  \frac{1}{k}\cdot k^{r-k}e^{-ak}}{\sum_{k=\frac{r}{s}}^r k^{r-k}e^{-ak}}.
\]
Therefore \Cref{expression} is the same as $q\E{1/x}$ when $x$ is a nonnegative integer drawn from the discrete distribution with density proportional to $f(x)=x^{r-x}e^{-ax}$. We let $z$ be as in \Cref{z} of \Cref{helper}, so $z \approx \frac{r}{sq}$. We use \Cref{helper} 
to bound
 \begin{align*}
 \E{\left| \frac{1}{x}-\frac{1}{z} \right|} 
 &\leq  \sum_{t=1}^\infty  \left(\frac{1}{z}-\frac{1}{z+ t \sqrt{z} }\right)e^{-t} +\sum_{t=1}^{\sqrt{z}-1} \left(\frac{1}{z- t \sqrt{z} }- \frac{1}{z}\right)e^{-t}  \\
  &=  \sum_{t=1}^\infty \frac{t\sqrt{z}e^{-t}}{z(z+t\sqrt{z})}+\sum_{t=1}^{\sqrt{z}-1} \frac{t\sqrt{z}e^{-t}}{z(z-t\sqrt{z})}\\
    &\leq  \frac{1}{z} \sum_{t=1}^\infty \frac{te^{-t}}{\sqrt{z}+1}+\frac{\sqrt{z}}{z} \brac{ \sum_{t=1}^{\sqrt{z}/3} \frac{3te^{-t}}{2z} +\sum_{t=\sqrt{z}/3}^{\sqrt{z}-1} te^{-t}  }\\
  &=  \frac{O(1)}{z\sqrt{z} } + \frac{O(1)}{ z\sqrt{z}} + O\brac{ \frac{\sqrt{z}}{3} e^{-\frac{\sqrt{z}}{3} } } = \frac{O(1)}{z \sqrt{z}}.
\end{align*}
Using this and approximating $z$ by $\frac{r}{sq}$, the expectation of $x$ with respect to the density proportional to $f$ can be estimated:
\[
q \E{\frac{1}{x}}=\frac{q}{z}\left(1+O\left(\frac{1}{\sqrt{z}}\right)\right)=\lo\frac{sq^2}{r}\left(1+O\left(\sqrt{\frac{sq}{r}}\right)\right)=(1+o_r(1))\frac{sq^2}{r}
\]
as claimed.


	 \end{proof}
	 
	 \section{Section \ref{extension} proofs}\label{ex proofs}

	 \begin{proof} (of \Cref{ex cc}.) Let $v$ be a vertex with target degree $t= t(v)$, and let $k$ denote the number communities containing $v$.  First we claim $deg(v)\sim Bin \left((s-1)k, \frac{t q}{s} \right)$.  Let $s$ be an arbitrary vertex of a community $S$ containing $v$. 
	 	 $$\Pr{ s \sim v \text{ in } S} = \sum_{u \in V} \Pr{s=u} \Pr{v \sim u \text{ in } S} = \sum_{u \in V} \frac{1}{n} \frac{ t(u) t q}{ds}= \frac{tq}{s}.$$
	 	 A vertex in $k$ communities has the potential to be adjacent to $(s-1)k$ other vertices, and each adjacency occurs with probability $t q /s$.
	 	 
	 	 Next, let $N_u$ be the event that a randomly selected neighbor of vertex $v$ is vertex $u$. We compute
	 	 \begin{align}
	 	 \Pr{N_u}&= \sum_{r} \frac{\Pr{ u \sim v \given deg(v)=r}\Pr{deg(v)=r}}{r}\nonumber \\
	 	 &=\sum_r \frac{ \Pr{u \sim v} \Pr{deg(v)=r \given u \sim v}}{r}\nonumber\\
	 	 &=\Pr{u \sim v} \E { \frac{1}{deg(v)} \given u \sim v}\nonumber\\
	 	 &= \lo  \bfrac{s}{n}^2\frac{n}{(s-1)q} \frac{ t(u)t q}{sd} \left(\frac{1-e^{-tqk}}{tk q}\right)\label{deg ex} \\
	 	 &= \lo  \frac{t(u)\brac{1-e^{-tqk}}}{qkdn}\nonumber.
	 	 \end{align}
	 	To see \Cref{deg ex}, note that by the first claim $\E { \frac{1}{deg(v)} \given u \sim v}= \E{\frac{1}{X+1}}$ where $X \sim Bin\left((s-1)k -1, \frac{tq }{s}\right)$. Applying \Cref{ex of rec p1} and assuming $s=\omega(1)$, we obtain 
	 	$$\E { \frac{1}{deg(v)} \given u \sim v}=\frac{1-(1-\frac{tq }{s})^{(s-1)k }}{((s-1)k )\frac{tq }{s}}= \lo \frac{1-e^{-tqk}}{tk q}.$$
	 	
Now we compute the expected clustering coefficient conditioned on the number of communities the vertex is part of under the assumption that $s/n=o(q)$. Observe \begin{align}
	 	\E {C(v) \given \text{$v$ in $k$ communities}}&=\sum_{u,w} N_u N_w\Pr{ u \sim w \given u \sim v \text{ and } w \sim v }\nonumber \\
	 	&=\sum_{u,w} \frac{t(u)t(w)\brac{1-e^{-tqk}}^2}{(qkdn)^2}\brac{ \frac{1}{k} + \bfrac{s}{n}^2\frac{n}{(s-1)q} } \frac{t(u)t(w)q}{sd}\nonumber \\
	 	&= \lo \frac{\brac{1-e^{-tqk}}^2 \brac{\sum_{u \in V} t(u)^2}^2}{qd^3k^3n^2s}. \label{cc droc}
	 	\end{align}

Next compute the expected clustering coefficient without conditioning on the number of communities. To do so we need to compute the expected value of the function $f(k)=\frac{(1- e^{-kqt})^2}{k^3}$. We first use Taylor's theorem to give bounds on $f(k)$. For all $k$, there exists some $z\in [1/q, k]$ such that
\begin{equation*}
f(k)= f\bfrac{1}{q} + f'\bfrac{1}{q} \brac{k-\frac{1}{q}} + \frac{ f''(z)}{2}\brac{k-\frac{1}{q}}^2.
\end{equation*}
Note that for $z \in [1/q, k]$ 
\begin{align*} f''(z)&= \frac{12 (1 - e^{-k q t})^2}{k^5} - \frac{
 12 e^{-k q t}(1 - e^{-k q t}) q t}{k^4} + \frac{
 2 e^{-2 k q t} q^2 t^2}{k^3} - \frac{
 2 e^{-k q t} (1 - e^{-k q t}) q^2 t^2}{k^3}\\
 &\leq \frac{12 (1 - e^{-k q t})^2}{k^5} + \frac{
 2 e^{-2 k q t} q^2 t^2}{k^3}\\
 &\leq q^5\brac{ 12+ 2t^2e^{-2t}},  
 \end{align*}
and $$f''(z) \geq 0.$$
It follows that 
\begin{equation} \label{taylor} 
f\bfrac{1}{q} + f'\bfrac{1}{q} \brac{k-\frac{1}{q}}  \leq f(k)\leq  f\bfrac{1}{q} + f'\bfrac{1}{q} \brac{k-\frac{1}{q}} + q^5\brac{ 6+ t^2e^{-2t}}\brac{k-\frac{1}{q}}^2.
\end{equation}

 Let $M \sim Bin(n/(sq), s/n)$ be the random variable for the number of communities a vertex $v$ is part of. (Since $s= \omega(1)$ replacing the number of communities by $n/(sq)$ changes the result by a factor of $\lo$.) We use \Cref{taylor} to give bounds on the expectation of $f(M)$, 
 
 \begin{align*}
 \E{f(M)} &\leq  \E{ f\bfrac{1}{q} + f'\bfrac{1}{q} \brac{M-\frac{1}{q}} + q^5\brac{ 12+ 2t^2e^{-2t}}\brac{M-\frac{1}{q}}^2}\\
 &=(1- e^{-t})^2q^3 + \frac{1}{q}\brac{1-\frac{s}{n}} q^5 \brac{ 6+ t^2e^{-2t}}\\
 &\leq (1- e^{-t})^2q^3 + q^4 \brac{ 6+ t^2e^{-2t}}
 \end{align*}
 and 
 \begin{align*}
 \E{f(M)} \geq  \E{ f\bfrac{1}{q} + f'\bfrac{1}{q} \brac{M-\frac{1}{q}} }
 &=(1- e^{-t})^2q^3. \end{align*} 
Therefore $ \E{f(M)}  = (1- e^{-t})^2q^3 + c_t q^4 $ for some constant $c_t \in [0, 6.2)$.

Finally, we compute
		\begin{align*} 
		\E {C(v)}&= \sum_k \Pr{M=k}\frac{\brac{1-e^{-tqk}}^2 \brac{\sum_{u \in V} t(u)^2}^2}{qd^3k^3n^2s}\\
		&= \frac{\brac{\sum_{u \in V} t(u)^2}^2}{qd^3n^2s} \E {f(M)}\\
		&= \lo \frac{\brac{\sum_{u \in V} t(u)^2}^2}{d^3n^2s} \brac{ (1- e^{-t})^2q^2 +c_tq^3}.
		\end{align*}

%
	 	\end{proof}

\begin{proof} (of \Cref{power law}.) Let $d= mean(D)$. We compute $$ \E {d} = \sum_{k=1}^\infty \frac{k^{-\gamma+1}}{\zeta(\gamma)} = \frac{\zeta(\gamma-1)}{\zeta(\gamma)}.$$
	
	Next we claim that with high probability the maximum target degree of a vertex is at most $t_0=n^{2/(\gamma-1)}$. Let $X$ be the random variable for the number of indices $i$ with $t(v_i)>k_0$. 
	
	\begin{align*}
	\Pr{  \max_i t(v_i) > t_0 } &\leq \E {X }= n \Pr{ t(v_1) > t_0} \leq n \sum_{i=t_0+1}^{\infty} \frac{i^{-\gamma}}{\zeta(\gamma)}\\
	&\leq n\int_{i=t_0}^{\infty}  \frac{i^{-\gamma}}{\zeta(\gamma)}= \bfrac{1}{\zeta(\gamma)(\gamma-1)} n t_0^{1-\gamma}=o(1).
	\end{align*}
	It follows that $ \max_i t(v_i)^2 \leq n^{\frac{1}{\gamma-1}}$, and so $\max_i t(v_i)^2 \leq \frac{sd}{q}$.
\end{proof}

	 \end{document}